\documentclass[12pt]{article}

\usepackage{amsthm,amsmath,amssymb,bm,natbib,tikz,mathtools,placeins,subcaption,graphicx,fullpage}

\usetikzlibrary{shapes,backgrounds}

\newtheorem{lemma}{Property}

\DeclareMathOperator*{\argmax}{arg\,max}

\begin{document}

\title{Towards a Multi-Subject Analysis of Neural Connectivity}
\date{}
\author{C. J. Oates$^1$\footnote{E-mail: c.oates@warwick.ac.uk}, L. Costa$^1$\footnote{E-mail: l.c.carneiro-da-costa@warwick.ac.uk}, T. E. Nichols$^{1,2}$\footnote{E-mail: t.e.nichols@warwick.ac.uk} \\ $^1$Department of Statistics and $^2$Warwick Manufacturing Group,\\ University of Warwick, Coventry, CV4 7AL, UK.}
\maketitle

\begin{abstract} 
Directed acyclic graphs (DAGs) and associated probability models are widely used to model neural connectivity and communication channels. 
In many experiments, data are collected from multiple subjects whose connectivities may differ but are likely to share many features. 
In such circumstances it is natural to leverage similarity between subjects to improve statistical efficiency.
The first exact algorithm for estimation of multiple related DAGs was recently proposed by \cite{Oates6}; in this letter we present examples and discuss implications of the methodology as applied to the analysis of fMRI data from a multi-subject experiment.
Elicitation of tuning parameters requires care and we illustrate how this may proceed retrospectively based on technical replicate data.
In addition to joint learning of subject-specific connectivity, we allow for heterogeneous collections of subjects and simultaneously estimate relationships between the subjects themselves.
This letter aims to highlight the potential for exact estimation in the multi-subject setting.
\end{abstract}

\section{Introduction}

Probabilistic graphical models are widely used to model neural connectivity and the transfer of information between regions of the brain \citep{Poldrack}.  
In brief, vertices indexed by $1,\dots,P$ in a directed acyclic graph (DAG) $G$ are identified with random variables $Y_i$ that represent neural activity at a particular region and edges between the vertices describe conditional independence statements, whose interpretation depends on both the underlying statistical model for the data and the context in which data are obtained.  
In many neuroscience applications, subject-specific connectivity (i.e. the set of edges) itself is uncertain and an important challenge is to infer this structure from experimental data \citep{Friston}.
There has been considerable statistical research into inference for graphical models in general over the last decade, with particular emphasis on Bayesian networks \citep[BNs;][]{Chickering,Friedman,Ellis}, Gaussian graphical models \citep[GGMs;][]{Meinshausen,Chandrasekaran} and discrete graphical models \citep{Loh}.  Nevertheless there remain two substantive barriers to the inference of graphical models from data:
Firstly, inferred graphical structure is often not robust to reasonable perturbation of the underlying data \citep{Claassen}. This is due to a combination of the high variance of graphical estimators themselves and any additional variance that is introduced if the structure learning algorithm returns only an approximation to the intended estimator. 
Secondly, conventional model selection criteria for graphical models are often biased towards selecting more complex models (i.e. more edges), since there are typically very many models in which the data-generating model is nested; these models are also able to fit the data well \citep[albeit with some coefficients close or equal to zero;][]{Consonni}. 
Consequently many more data are required to exclude more complex alternatives.
Taken together, these factors limit the extent to which neural connectivity can be accurately recovered from data.

Many experimental designs in neuroscience involve data collected on multiple subjects, indexed by $1,\dots,K$, that may differ with respect to neural connectivity, such that corresponding graphs $G^{(k)}$ may be subject-specific \citep{Sugihara,Li}.
Efforts to analyse multi-subject experimental data have previously focussed on hierarchical models and imaging data, rather then connectivity {\it per se} \citep{Mumford,Sanyal,Badillo,Marquand}.
Given that elements of neural architecture are largely conserved between subjects, it is natural to leverage this similarity in order to improve statistical efficiency, by addressing both the robustness of inferred graphical structure and reducing small sample bias \citep{Mechelli}.
The statistical chellenge of estimating multiple related graphical models has recently received much attention:
For GGMs, \cite{Danaher} and others exploited $L_1$ penalties, such as the fused graphical LASSO, to couple together inference for multiple related subjects. 
Such penalised likelihood methods are computationally tractable and scale well to high dimensions. 
These studies demonstrate that it is possible to increase statistical efficiency, often considerably, by formulating an appropriate joint model that couples together multiple graphs.
Likewise, the methodology improves robustness by requiring that graphical structure is approximately invariant to perturbations of the data that are, in effect, provided by the subjects themselves.

Whilst useful in many applications, GGMs are undirected graphs and hence cannot not represent the direction of information flow between neural regions.
More fundamentally, GGMs do permit causal inference that is typically the scientific objective \citep{Valdes}.
For this reason we focus attention on graphical models, such as BNs, that are based on DAGs and have an associated theory of inferred causation \citep{Pearl}.
Research focussing on DAGs in this setting includes \cite{Ramsey}, who constructed a hierarchical model in which graph structure was conserved between subjects but the parameters that describe the data-generating process were subject-specific.
\cite{Waldorp} went further by permitting subject-specific graph structure and parameters in the context of Gaussian ancestral graphs whose parameters are constrained by a hierarchical model.
This latter work is closest in spirit to the methodology that we discuss below, but we do not restrict attention to either stationary data or Gaussian data, rendering our approach considerably more flexible.
 
Until very recently, estimation of more general DAGs required either the strong assumption that an ordering of the variables $1,\dots,P$ applied equally to all subjects \citep{Oyen}, or the use of expensive computational approximations such as Markov chain Monte Carlo that scale extremely poorly as either the number $P$ of variables or number $K$ of subjects grows \citep{Werhli}.
An exact algorithm that facilitates the joint estimation of multiple DAGs was recently developed in the sister paper \cite{Oates6}, viewing the estimation problem within a hierarchical Bayesian framework (somewhat similar to a random effects model for the graph structure) and applying advanced techniques from integer linear programming to obtain a {\it maximum a posteriori} estimate of all DAGs simultaneously.
The availability of exact algorithms offers the opportunity to analyse multi-subject neural connectivity using causal DAG models, whilst leveraging the similarity between subjects in order to improve statistical efficiency and robustness.
This letter illustrates the scope and applicability of these exact algorithms within neuroscience, using a small functional magnetic resonance imaging (fMRI) time course dataset obtained on six subjects, coupled with multiregression dynamical models \citep[MDMs;][]{Queen} that permit statistically rigorous causal inference \citep{Queen2}.
It is envisaged that exact algorithms will play an important r$\hat{\text{o}}$le in future studies of neural connectivity and this letter serves to illustrate their application by example.

\begin{figure}[t!]
\centering
\begin{subfigure}[b]{0.7\textwidth}
\includegraphics[width = \textwidth,clip,trim = 1cm 0cm 1cm 0cm]{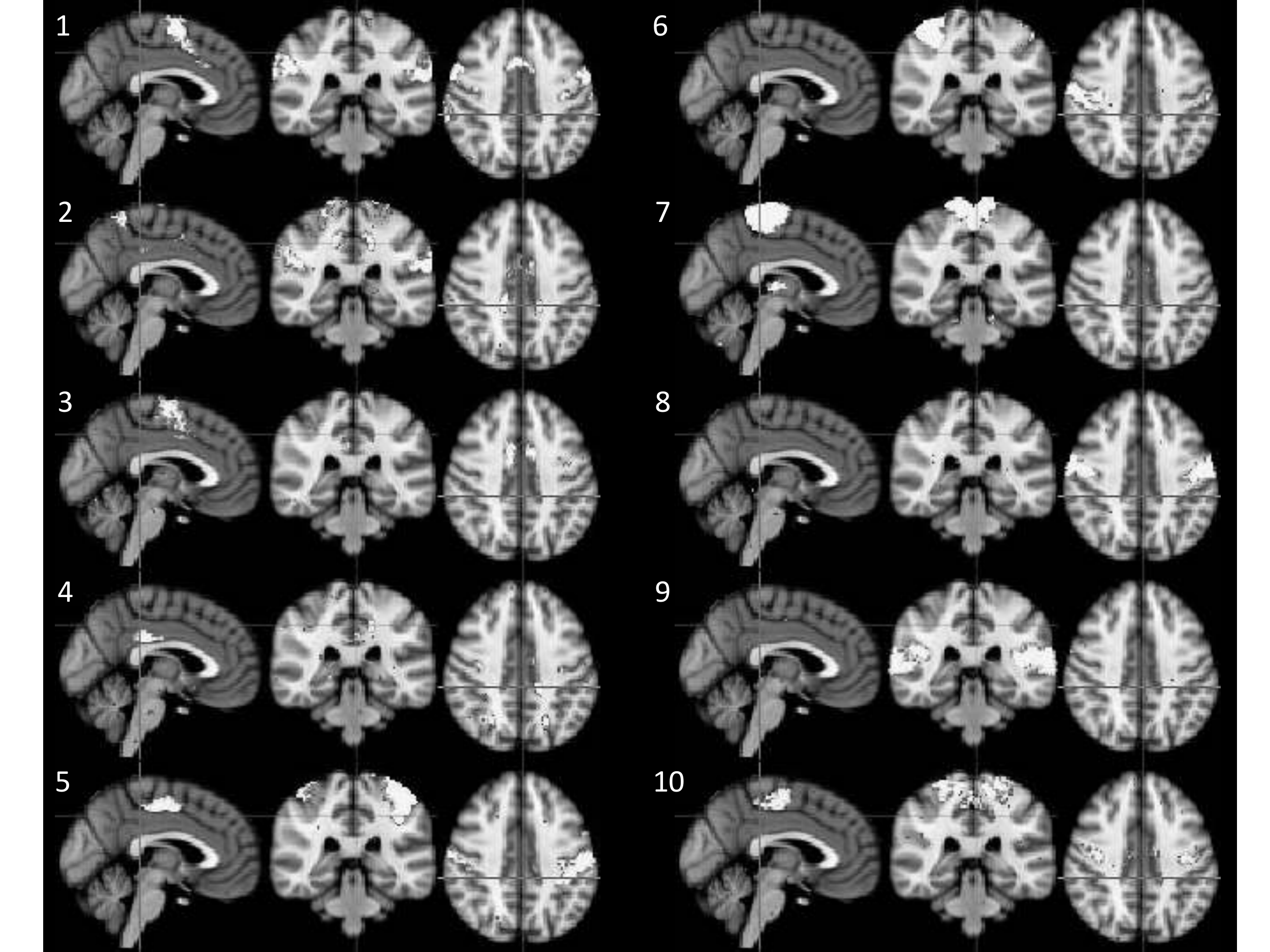}
\caption{Estimated ICA components}
\label{ICA modes}
\end{subfigure}

\begin{subfigure}[b]{0.55\textwidth} \label{regions}
\resizebox{\textwidth}{!}{
\begin{tabular}{|ccc|}
\hline
\underline{Node Number} & \underline{Symmetry} & \underline{Summary} \\
1 &         Bilateral           &          Motor:hand/face \\
2  &        Bilateral            &         Sensory:All-but-face \\
3   &       Bilateral             &        Motor:All-but-face \\
4    &      Bilateral              &       UNKNOWN \\
5    &      Left Dominant      &    Sensorimotor: L Hand+Arms \\
6    &      Right Dominant    &    Sensorimotor: R Hand+Arms \\
7    &      Bilateral              &       Sensory: Trunk-to-feet \\
8    &      Bilateral              &       Sensory: Face \\
9    &      Bilateral              &       Auditory \\
10  &     Bilateral               &      Sensorimotor:All-but-face - Sensory:Face \\ \hline
\end{tabular}}
\caption{Functionality of neural regions}
\label{tab}
\end{subfigure}
\begin{subfigure}[b]{0.225\textwidth}
\includegraphics[width = \textwidth,clip,trim = 4cm 18.75cm 12.5cm 4cm]{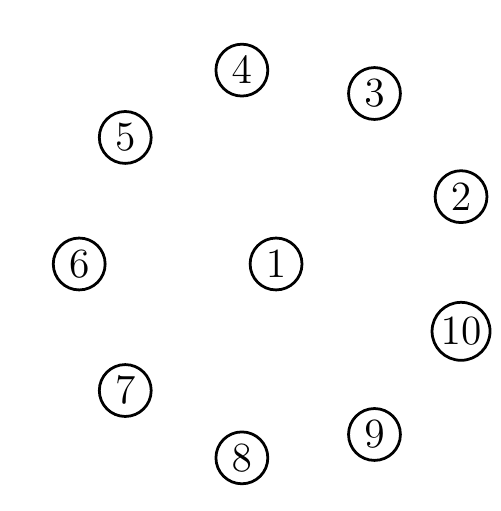}
\caption{Vertex layout}
\label{keys}
\end{subfigure}
\caption{Illustrative resting-state fMRI dataset. We consider 10 spatial modes obtained using ICA, show in (a) and described in (b), each having a corresponding time series for each subject. All graphs that we present will adopt the vertex layout shown in (c).
[A colour version of this figure is available in the online supplemental material.]}
\end{figure}

\section{Results}

\subsection{fMRI data and experimental setup}

Exact algorithms are illustrated here with a small fMRI dataset consisting of six subjects from the Human Connectome Project \citep{VanEssen}.  Scans were acquired on each subject while they were in a state of quiet repose; data from one 15 minute session were used, with a spatial resolution of $2 \times 2 \times 2$ mm and a temporal resolution of 0.7 secs; see \cite{Smith} for full details.  After correcting for head motion, all data was registered to a common reference atlas space and 100-dimensional independent component analysis (ICA) was conducted on the temporally concatenated data. The result of this ICA was 100 spatial modes (common to all subjects) and 100 corresponding temporal modes (subject-specific); at this high dimension, the 100 spatial modes are sparse and spatially compact (though possibly bilaterally symmetric) and so essentially provide a data-driven parcellation of the brain.  
Hierarchical clustering was used on the ICA temporal modes following \cite{Beckmann}, and the 10-mode cluster corresponding to motor cortex was selected for study here.  
Thus our data consists of 10 nodes, with a time series for each node for each subject.  
Figure \ref{ICA modes} displays the neural regions that we consider and Figure \ref{tab} shows the approximate description of each region; note that region 4 was spatially diffuse and difficult to characterise, and thus is likely to be an artefactual component.

The goal here is to understand neural information transfer at resting state and establish subject-specific connectivity.
By its very nature, estimation of resting state connectivity is challenging due to limited information content in the fMRI time series.
Indeed, \cite{Smith2} reported that whilst the presence or absence of connections can sometimes be estimated from fMRI time series data, estimating the direction of edges from data remains extremely challenging.
The integration of data from multiple related subjects offers one route to increased statistical power and this is the approach that we pursue here.

\subsection{MDMs for fMRI time series data}

Following data preprocessing we are left with a collection of random variables $Y_i^{(k)}(n)$ representing the observed activity in subject $k$ at region $i$ and time point $n$.
Following recent research by \cite{Costa} into causal inference based on such fMRI time course data, we model the $Y_i^{(k)}(n)$ as arising from a causal MDM.
Specifically, an MDM is defined on a multivariate time series $\bm{Y}^{(k)}$ is characterised by a contemporaneous DAG $G^{(k)}$, with information shared across time through evolution of the model parameters $\bm{\theta}_i^{(k)}(n)$.
We consider the case where $\bm{Y}^{(k)}(n) | \bm{\theta}^{(k)}(n)$ satisfies linear Gaussian structural equations, though any formulation would be compatible with the methodology that we present.
Write $G_i^{(k)} \subseteq \{1,\dots,P\}\setminus\{i\}$ for the parents of vertex $i$ in the DAG $G^{(k)}$ and write $\bm{Y}_S^{(k)}$ for the collection of univariate time series corresponding to the variables $\{Y_i^{(k)} : i \in S\}$.
This MDM is described by the following observation equations
\begin{eqnarray}
Y_i^{(k)}(n) = \mathbf{Y}_{G_i^{(k)}}^{(k)}(n)^T \boldsymbol{\theta}_i^{(k)}(n) + \epsilon_i^{(k)}(n)    \label{mdmeq1} 
\end{eqnarray}
\noindent where $\epsilon_i^{(k)}(n) \sim N(0,V_ i^{(k)}(n))$, together with the system equations
\begin{eqnarray}
\boldsymbol{\theta}^{(k)}(n) = \bm{\Gamma}^{(k)}(n) \boldsymbol{\theta}^{(k)} (n-1) + \mathbf{w}^{(k)}(n)  \label{mdmeq2}
\end{eqnarray}
where $\bm{\Gamma}^{(k)}(n)$ is a matrix of autoregressive coefficients and $\mathbf{w}^{(k)}(n)\sim N (\mathbf{0},\mathbf{W}^{(k)}(n))$.
Default choices for $V_i^{(k)}(n)$, $\bm{\Gamma}^{(k)}(n)$, $\bm{W}^{(k)}(n)$ were assumed following \cite{Costa}.
Model selection for MDMs is based on Bayes factors \citep[see e.g.][]{West}. 
The evidence in favour of the DAG $G^{(k)}$ under the MDM likelihood can be calculated as
\begin{eqnarray}                            
p(\bm{Y}^{(k)}|G^{(k)}) = \prod_{i=1}^{P} \prod_{n=1}^N p(Y_i^{(k)}(n) | \mathbf{Y}_{G_i^{(k)}}^{(k)}(n), \bm{Y}^{(k)}(1:n-1), G_i^{(k)}). \label{mdm evidence}
\end{eqnarray} 
In practice Eqn. \ref{mdm evidence} is evaluated using simple Kalman filter recurrences and we refer the reader to \cite{Costa} for further details.
\cite{Costa2} reports that the MDMs above are well-suited to the analysis of resting-state fMRI data, outperforming the methods surveyed by \cite{Smith2} in both the detection of edges and also the orientation of edges.
This promising performance appears to be driven by the information present in temporal spike patterns, as exploited directly in recent work by \cite{Diekman}.

The MDMs here are reified with the interpretation that edges correspond to neural connectivity \citep{Dawid2}.
Independent estimation for the subject-specific DAGs $G^{(k)}$ based on the MDM score (Eqn. \ref{mdm evidence}) yields graphs that display high between-subject variability (Fig. \ref{independent_neuro}).
Thus the causal semantics that are associated with MDMs imply that neural connectivity is highly variable between subjects.
This is unreasonable on neuroscientific grounds and likely reflects the lack-of-robustness and small sample bias that are often associated with graphical analyses. 
This motivates a hierarchical statistical model and exact estimation, as we describe below.

\begin{figure}
\centering
\begin{subfigure}[b]{0.8\textwidth}
\includegraphics[width = \textwidth,clip,trim = 3.6cm 17.8cm 16.5cm 2.5cm]{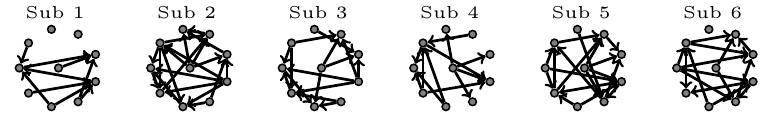}
\caption{}
\label{independent_neuro}
\end{subfigure}
\begin{subfigure}[b]{0.8\textwidth}
\includegraphics[width = \textwidth,clip,trim = 3.6cm 17.8cm 16.5cm 2.5cm]{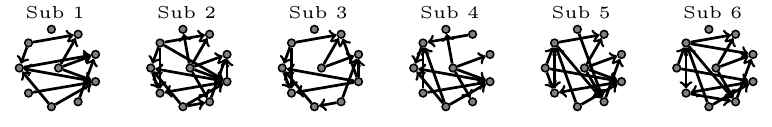}
\caption{}
\label{joint_neuro}
\end{subfigure}
\caption{Illustrative fMRI dataset; here time course data were obtained on six subjects.
The subject-specific connectivity between ten distinct regions of the brain was estimated using multiregression dynamical models applied (a) to each subject separately, (b) to all subjects jointly with regularity hyperparameter $\lambda = 4$. 
The graphs in (b) are 23\% more similar compared to the graphs in (a), as explained in the main text. 
[Figure \ref{keys} provides a key.]}
\end{figure}

\subsection{A hierarchical model for multi-subject analysis}

Following unsatisfactory independent estimation, we now proceed to explore exact joint estimation as enabled by the recent methodological advances of \cite{Oates6}.
We note that, providing that the quantities used to compute Eqn. \ref{mdm evidence} above have been cached, the joint analysis below does not require any further computation involving the MDM model equations.
Write $\mathcal{G}$ for the collection of all DAGs on the vertices $1,\dots,P$ and write $G^{(1:K)} \in \mathcal{G}^K$ for the collection of all the DAGs $G^{(1)}, \dots, G^{(K)}$.
Joint estimation proceeds within a hierarchical Bayesian framework that is specified by the ``multiple DAG prior''
\begin{eqnarray}
p(G^{(1:K)}|A) \propto \left( \prod_{(k,l) \in A} r(G^{(k)},G^{(l)}) \right) \times \left( \prod_{k = 1}^K m(G^{(k)}) \right). \label{joint prior}
\end{eqnarray}
The functions $r$ and $m$ are defined below.
Here $A$ denotes an undirected network on the indices $1,\dots,K$ that will be used to encode a similarity structure between subjects; the first product factorises along the edges of $A$.
When $A$ is complete, Eqn. \ref{joint prior} encodes an exchangeability assumption that any DAG $G^{(k)}$ is equally likely {\it a priori} to be similar to any other DAG $G^{(l)}$ ($k \neq l$).
Such an exchangeability assumption is implicit in much of the recent literature on multiple graphical models \citep{Werhli,Oyen,Danaher}.
However, exchangeability will be inappropriate when the collection of subjects is heterogeneous, for instance containing groups or subgroups that correspond to differential neural connectivities.
The methodology that we present below allows for arbitrary (and even uncertain) $A$, relaxing this exchangeability assumption and permitting more flexible estimation.

The function $r : \mathcal{G} \times \mathcal{G} \rightarrow [0,\infty)$ is used to encode regularity between pairs of DAGs, with larger values corresponding to {\it a priori} more similar DAG structures.
\cite{Oates6} showed that a particularly convenient form of regularity function is obtained by considering hyper-Markov properties \citep{Dawid}:
\begin{eqnarray}
\log(r(G^{(k)},G^{(l)})) = - \sum_{i=1}^P \sum_{j=1}^P \lambda_{j,i}^{(k,l)} [ (j \in G_i^{(k)}) \oplus (j \in G_i^{(l)}) ]. \label{hm}
\end{eqnarray}
Here $\oplus$ is the logical XOR operator and $[E]$ is used to denote an indicator function for the event $E$.
The constants $\lambda_{j,i}^{(k,l)}$ can be used to encode which aspects of structure are more likely to be conserved across subjects, based on subjective prior information, or indeed to encode which subjects are more likely to share similar connectivity, based on ancillary covariates such as age, gender, disease status etc.
For example one could exploit a penalty $\lambda_{j,i}^{(k,l)} = \lambda^{-|\text{age}(k) - \text{age}(l)|}$ for some $\lambda > 0$ that encourages sharing of graph structure among subjects of similar ages, but treats edges $(j,i)$ as exchangeable.
The function $m(G)$ in the multiple graphical model prior (Eqn. \ref{joint prior}) provides an adjustment for the fact that the size of the space $\mathcal{G}$ grows super-exponentially with the number $P$ of vertices \citep{Consonni}.
In this paper we follow \cite{Chen,Scott,Foygel} and control multiplicity using the binomial correction
\begin{eqnarray}
m(G) = \prod_{i = 1}^P \binom{P}{|G_i|}^{-1} [ |G_i| \leq d_{\max} ].
\end{eqnarray}
Here $d_{\max}$ is a fixed upper bound on the in-degree of vertices in $G$ that encodes prior knowledge on the support of the graphical models \citep[e.g.][]{Hill}.
For all examples in this paper we made the subjective choice $d_{\max} = 3$ that reflects the degree of connectivity observed in previous literature \citep[e.g.][]{Ramsey}.

\FloatBarrier
\subsection{Exact estimation of graphical structure}

Bayesian estimation of graphical structure is based on the {\it maximum a posteriori} (MAP) estimate that is obtained jointly over all DAGs as
\begin{eqnarray}
\hat{G}^{(1:K)}|A := \argmax_{G^{(1:K)} \in \mathcal{G}^K} p(G^{(1:K)}|\bm{Y}^{(1:K)},A). \label{opt1}
\end{eqnarray}

More generally, the network $A$ that expresses similarity between the subjects may be subject to uncertainty.
Write $\mathcal{A}$ for the set of undirected networks on the vertices $1,\dots,K$.
In this setting we impose a hyperprior distribution over $\mathcal{A}$ given by
\begin{eqnarray}
\log(p(A)) \stackrel{+C}{=} \sum_{k=1}^K \sum_{l=k+1}^K \eta^{(k,l)} [ (k,l) \in A ] \label{hyp}
\end{eqnarray}
where $\stackrel{+C}{=}$ denotes equality up to an unspecified additive constant.
Here constants $\eta^{(k,l)}$ can again be used to encode prior similarity between subjects on the basis of ancillary covariates.
The hyperprior distribution in Eqn. \ref{hyp} has the effect of detering sparsity in the network $A$, leading to increased regularisation between DAGs and a more conservative estimate of between-subject variability.
In this extended setting, our focus is now an extended MAP estimator
\begin{eqnarray}
(\hat{G}^{(1:K)},\hat{A}) := \argmax_{\substack{G^{(1:K)} \in \mathcal{G}^K \\ A \in \mathcal{A}}} p(G^{(1:K)},A|\bm{Y}^{(1:K)}) \label{opt2}
\end{eqnarray}
that simultaneously estimates subject-specific DAGs $G^{(k)}$ and the network $A$ that relates subjects.

For each of the optimisation problems in Eqns. \ref{opt1}, \ref{opt2}, \cite{Oates6} describes how techniques from integer linear programming, including constraint propagation and cutting planes, may be used to exactly locate the MAP estimate with a minimal computational burden.
Define the ``local scores''
\begin{eqnarray}
s^{(k)}(i,G_i^{(k)}) := \left\{ \begin{array}{ll} \log(p(\bm{Y}_i^{(k)}|\bm{Y}_{G_i^{(k)}}^{(k)},G_i^{(k)})) - \log \binom{P}{|G_i^{(k)}|} & \text{if }|G_i^{(k)}| \leq d_{\max} \\ -\infty & \text{otherwise} \end{array} \right. \label{scores}
\end{eqnarray}
that are sufficient statistics for structure learning in the MDM model. 
Specifically, $s^{(k)}(i,\pi)$ is equal (up to an additive constant) to the log-posterior probability associated with the parent set configuration $G_i^{(k)} = \pi$ for node $i$ in subject $k$, based only on the subject's own data $\bm{Y}^{(k)}$.
The joint MAP estimators introduced above seek to maximise the sum of these terms subject to the hierarchical penalty (Eqn. \ref{hm}) and the requirement that $G^{(1)}, \dots, G^{(K)}$ are each well-defined DAGs. 

In brief, the computational methodology exploits the fact that Eqns. \ref{opt1}, \ref{opt2} can be encoded as an integer linear program of the form
\begin{eqnarray}
\text{maximise} \; \; \; \bm{f}^T\bm{x} \; \; \; \text{subject to} \; \; \; \bm{A}\bm{x} \leq \bm{b}, \; \; \; \bm{C}\bm{x} = \bm{d}, \; \; \; \bm{x} \geq \bm{0}, \; \; \; x_i \in \mathbb{Z}
\end{eqnarray}
through careful choices of the integer-valued matrices $\bm{A}$, $\bm{C}$ and integer-valued vectors $\bm{b}$, $\bm{d}$.
Specifically, the entries of the vector $\bm{x}$ are taken to be binary indicators corresponding to events that include $[G_i^{(k)} = \pi]$, $[(k,l) \in A]$ and $[ (j \in G_i^{(k)}) \oplus (j \in G_i^{(l)}) ] \cap [(k,l) \in A]$, whilst the vector $\bm{f}$ contains the local scores $s^{(k)}(i,\pi)$ and the constants $\lambda_{j,i}^{(k,l)}$ and $\eta^{(k,l)}$.
By inspection of Eqns. \ref{mdm evidence}, \ref{hm} and \ref{hyp} we see that the posterior log-probability
\begin{eqnarray}
\log(p(G^{(1:K)},A|\bm{Y}^{(1:K)})) & \stackrel{+C}{=} & \sum_{k=1}^K \sum_{i=1}^P \sum_{\pi \subseteq \{1:P\}\setminus\{ i \}} s^{(k)}(i,\pi) [G_i^{(k)} = \pi]  + \sum_{k=1}^K \sum_{l=k+1}^K \eta^{(k,l)}[(k,l) \in A] \nonumber \\
&& - \sum_{k=1}^K \sum_{l=k+1}^K \sum_{i,j=1}^P \lambda_{j,i}^{(k,l)} [ (j \in G_i^{(k)}) \oplus (j \in G_i^{(l)}) ] \cap [(k,l) \in A] 
\end{eqnarray}
can be written as an inner-product $\bm{f}^T\bm{x}$.
The inequality constraints $\bm{A}\bm{x} \leq \bm{b}$ and equality constraints $\bm{C}\bm{x} = \bm{d}$ are carefully chosen to ensure that the feasible region for $\bm{x}$ consists of precisely those vectors that correspond to well-defined (multiple) DAG models.
This final point is somewhat technical and we refer the reader to \cite{Oates6} for full details.

\FloatBarrier
\subsection{Elicitation of tuning parameters; theory}

The class of statistical models that is amenable to exact inference is substantial, but here we focus on particularly tractable prior specifications that allows us to clearly illustrate the methodology.
Specifically, we reduce the number of hyperparameters to two by making the assumption that all edges are {\it a priori} equally likely to be shared between all pairs of subjects ($\lambda_{j,i}^{(k,l)} = \lambda$ for all $i,j,k,l$) and that all pairs of subjects are {\it a priori} equally likely to share similar graph structure ($\eta^{(k,l)} = \eta$ for all $k,l$).
Prior elicitation in this reduced class of models therefore requires the specification of hyperparameters $\lambda$ and $\eta$.
The impact of the choice of hyperparameters on the MAP estimators is clarified in the following:

\begin{lemma} \label{mix1}
(a) When $\eta = 0$, $\hat{G}^{(1:K)}$ consists of DAGs equal to those computed using independent estimation.
(b) For $\eta > 0$ we have $(k,l) \notin \hat{A} \implies \hat{G}^{(k)} \neq \hat{G}^{(l)}$.
(c) For fixed $\eta$ there exists $\lambda^* \in [0,\infty)$ such that whenever $\lambda > \lambda^*$ we have $(k,l) \in \hat{A} \implies \hat{G}^{(k)} = \hat{G}^{(l)}$. \label{mixture}
(d) There exists $\eta^* \in [0,\infty)$ such that $\hat{A}$ is the complete network whenever $\eta > \eta^*$.
\end{lemma}

The above result deals with the extremes of the parameter space; intuitively we would expect non-trivial values of $(\lambda,\eta)$ to interpolate ``smoothly'' between these extremes.
The following shows that this intuition is not strictly true.
Specifically, as $\lambda$ is monotonically increased, it is possible for a particular edge to enter and exit the MAP estimator $\hat{G}^{(k)}$ multiple times and furthermore, non-monotonicity is also exhibited by the network estimator $\hat{A}$:

\begin{lemma} \label{nonmono1}
(a) Fix a network $A \in \mathcal{A}$ and consider varying the hyperparameter $\lambda$. If $A$ is non-empty, then there exist values of the sufficient statistics $s^{(k)}(i,\pi)$ such that $[ j \in \hat{G}_i^{(k)} ]$ is not monotonic in $\lambda$. 
(b) Fix the hyperparameter $\lambda$, and consider unknown $A$ with hyperparameter $\eta$. Then there exist values of the sufficient statistics $s^{(k)}(i,\pi)$ such that $[ (j,i) \in \hat{A} ]$ is not monotonic in $\eta$. 
\end{lemma}
\noindent Thus the joint MAP, like other penalised likelihood approaches \citep[including the GLASSO for GGMs;][]{Friedman2} does not obey a monotonicity property.
Property \ref{nonmono1} makes it surprising that exact algorithms exists in this nontrivial setting.
In practice and in results below we have found that, like the GLASSO, monotonicity holds approximately.

\FloatBarrier

\begin{figure}[t!]
\centering
\begin{subfigure}[b]{0.9\textwidth}
\includegraphics[width = \textwidth,clip,trim = 3cm 13cm 12.8cm 2.5cm]{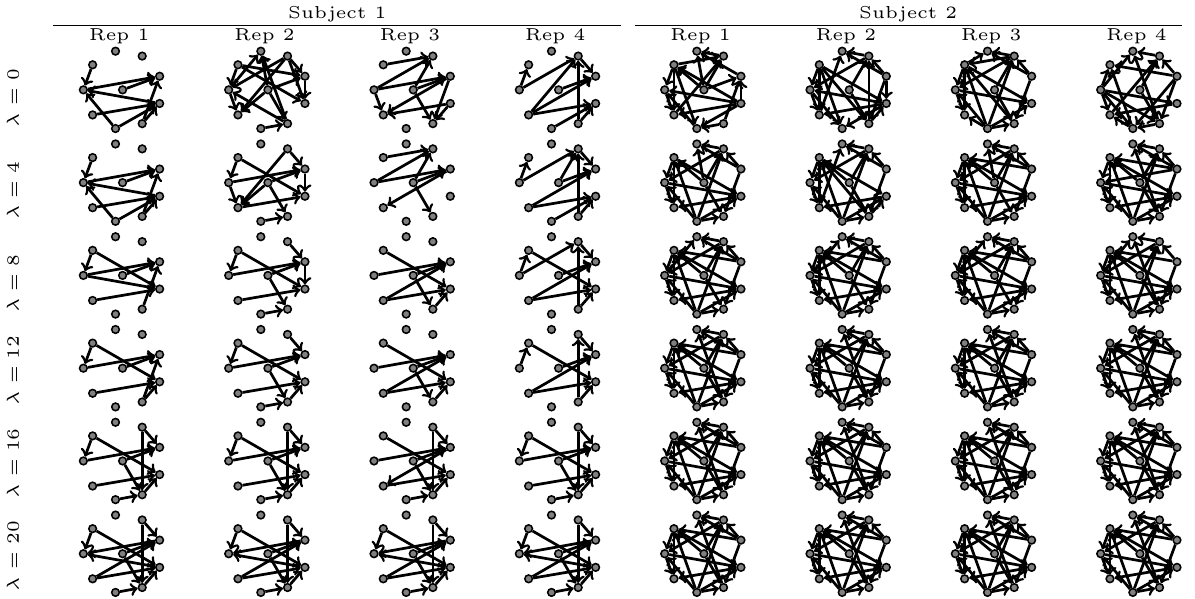}
\caption{}
\label{learn lambda}
\end{subfigure}
\begin{subfigure}[b]{0.49\textwidth}
\includegraphics[width = \textwidth,clip,trim = 3cm 8.6cm 4cm 9cm]{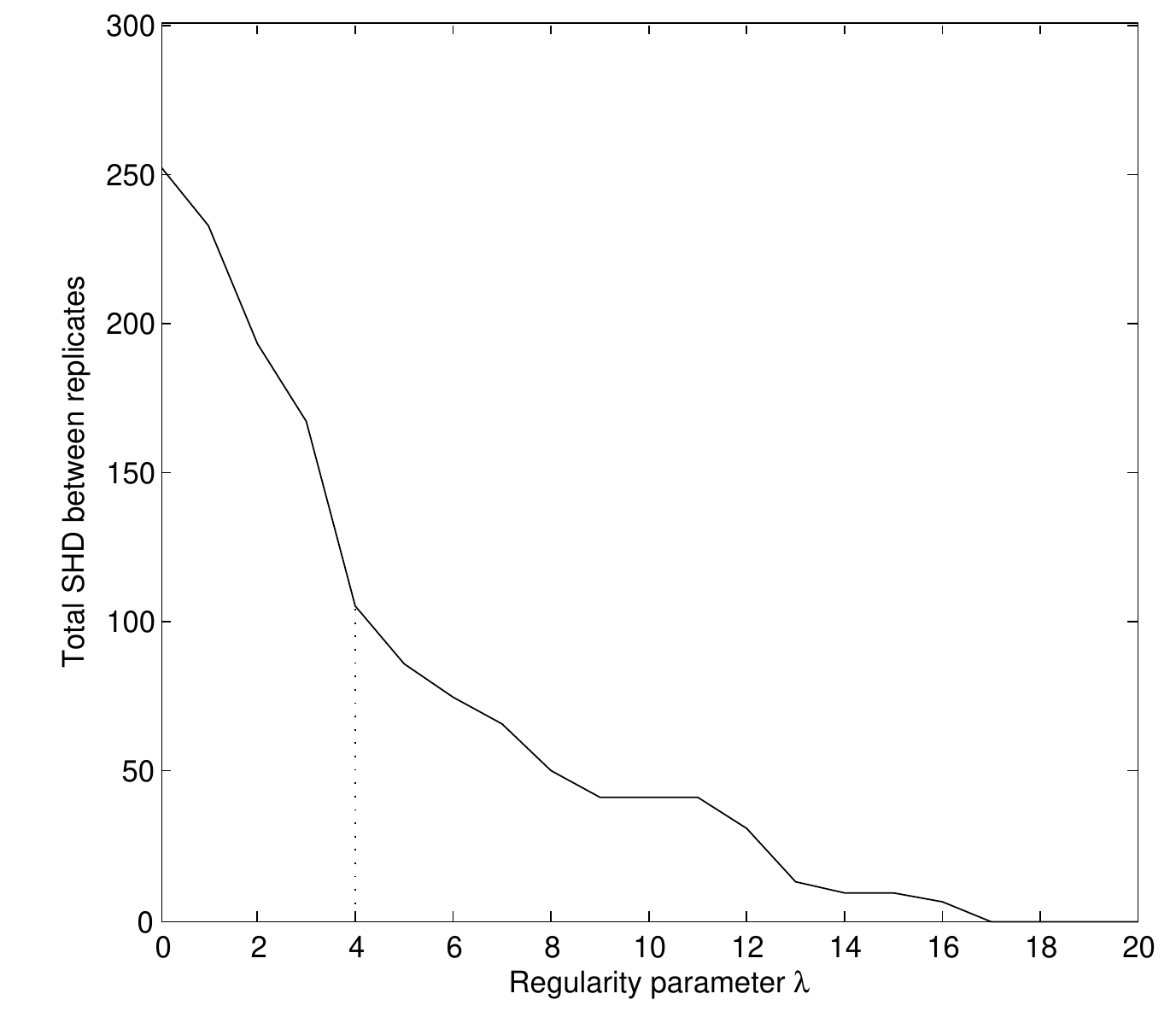}
\caption{}
\label{shda}
\end{subfigure}
\begin{subfigure}[b]{0.49\textwidth}
\includegraphics[width = \textwidth]{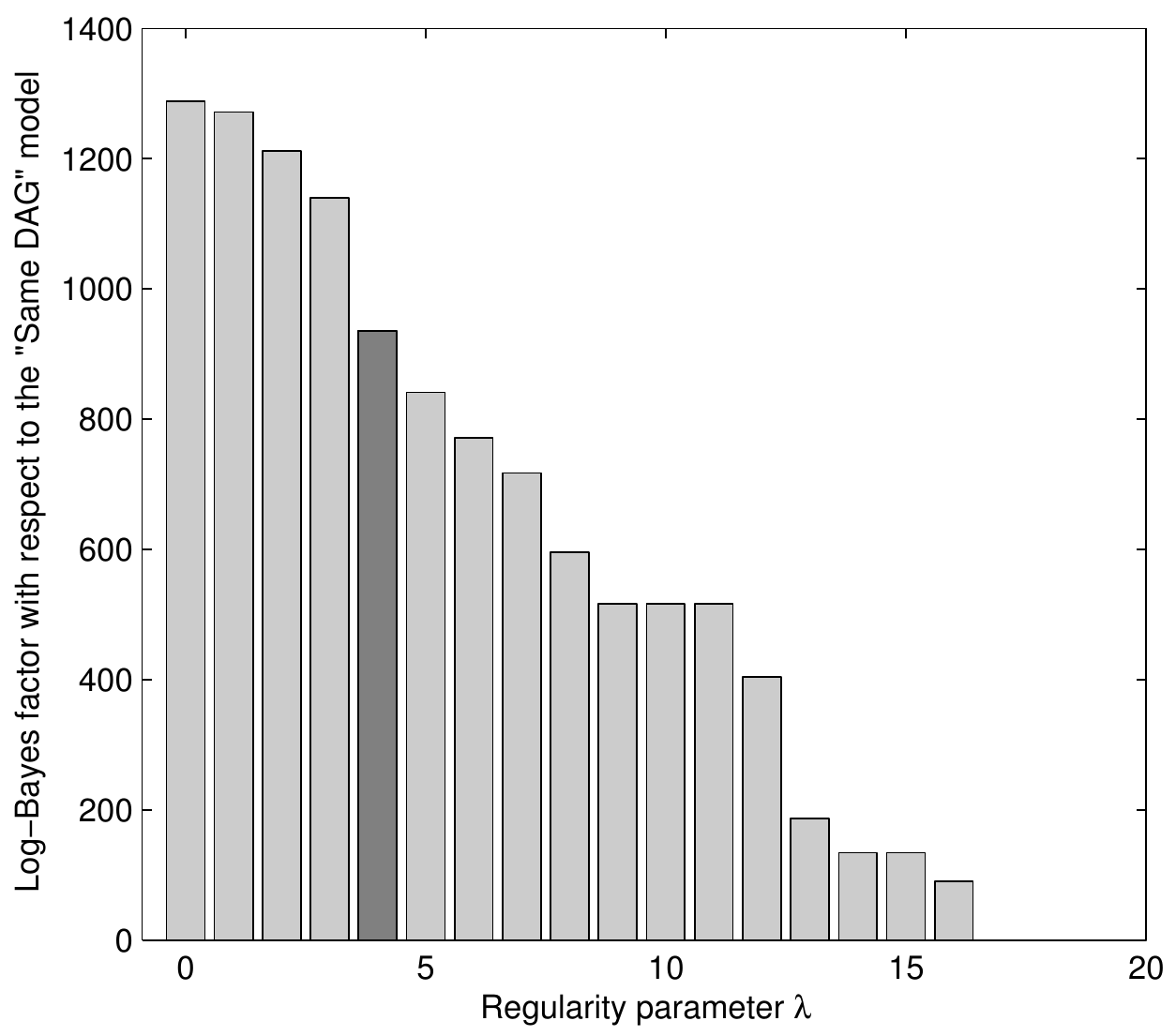} 
\caption{}
\label{BFs}
\end{subfigure}
\caption{Illustrative fMRI dataset; diagnostics for eliciting the regularity parameter $\lambda$ based on technical replicate data and retrospective inspection of the posterior. 
(a) Here two subjects each provided four technical replicate datasets. The DAGs shown are joint MAP estimates for varying $\lambda$, such that replicates were assumed to be exchangeable but subjects were treated independently.
As $\lambda$ is increased the DAGs corresponding to technical replicates become increasingly similar until they are identical at $\lambda \geq 17$.
(b) Here we plot the total SHD between DAGs corresponding to technical replicates against the regularity parameter $\lambda$. The dashed line indicates the value $\lambda = 4$ that reduces the between-subject variability by approximately 50\%.
(c) Comparing the Bayes factor corresponding to model $\lambda = 4$ against independent estimation ($\lambda = 0$) and estimation that forces all DAGs to be identical ($\lambda \geq 17$).}
\end{figure}

\subsection{Elicitation of tuning parameters; practice}

The elicitation of hyperparameters such as $\lambda$, $\eta$ should principally be driven by the scientific context, the nature of the data and the r\^{o}le that inferences are to play in future work.
For example, if the estimated networks are the basis for features within a classification algorithm, then elicitation of hyperparameters should target the classification error.
However in some settings, including our illustrative example, the non-availability of relevant ancillary data (e.g. the class labels in classification) precludes such an objective elicitation.
Below we therefore illustrate diagnostics that could form the basis for subjective elicitation in quite general settings, based on retrospective inspection of the posterior.

The analysis of resting state fMRI data is an emerging area of research \citep{Cole} and currently neither the source nor the extent of subject-specific variation are well-understood.
If the extent of variability at resting state was known, this could be directly leveraged to facilitate the objective elicitation of hyperparameters.
However this is not currently the case and subjective elicitation is required.
The biological knowledge that forms the basis for elicitation is qualitative in nature, as we explain below:
Firstly, connectivity should not change within a subject over the brief time period under which the fMRI experiments were conducted.
Secondly, recent studies \citep[e.g.][]{Ringach} indicate that the notion of ``resting state'' is poorly defined and can correspond to several contrasting neurological activity profiles; we would therefore not expect to obtain identical DAGs under a replication experiment that is unable to control for the precise nature of the resting state. 

A subjective analysis can be obtained using diagnostics based on retrospective examination of the posterior, that we describe below.
Specifically, for our fMRI dataset, we performed exact estimation of the joint MAP based on four technical replicate datasets obtained from the first two subjects under identical laboratory conditions.
To inform elicitation for the regularity parameter $\lambda$, we fixed the network $A$ such that $(k,l) \in A$ if and only if datasets $k$ and $l$ were both technical replicates derived from the same subject (Fig. \ref{learn lambda}). 
This corresponds to placing an exchangeability assumption on the technical replicates, but prohibiting the sharing of information between subjects.
We then computed the total structural Hamming distance \citep[SHD;][]{Tsamardinos} between all pairs of DAGs that are technical replicates (Fig. \ref{shda}).
This diagnostic could be used as the basis for subjective elicitation of $\lambda$ in general situations where replicate data are available.
Below for illustration we focus on one such value, $\lambda = 4$, that attributes approximately 50\% of variability between technical replicates to extrinsic noise resulting from the experimental design. 
Examination of the Bayes factor as a function of $\lambda$ provides a second diagnostic to assist with elicitation that may be useful to highlight over-regularisation.
In this case the value $\lambda = 4$ scores considerably better compared to the alternative that assigns the same DAG to all replicate datasets (log-Bayes factor $\approx 900$, Fig. \ref{BFs}). 
Additional diagnostics for the subjective elicitation of $\eta$ are discussed in the subsequent sections.

\FloatBarrier

\subsection{Learning multiple DAGs with exchangeability}

Based on the elicitation $\lambda = 4$, for illustration, we employed exact estimation for the joint MAP $\hat{G}^{(1:K)}|A$ under the exchangeability assumption that $A$ is the complete network (Eqn. \ref{opt1}).
In order to limit scope, here we simply consider one dataset per subject (i.e. no technical replicates were included), but data aggregation is naturally accommodated in the methodology we present (see discussion).
Results in Figure \ref{joint_neuro} demonstrate that the estimated DAG structures are substantially more similar that our original estimate obtained using independent inference (Fig. \ref{independent_neuro}), with a 23\% decrease in total SHD between DAGs.
This estimate can be expected to more closely represent the true subject-specific neural connectivity patterns, based on the empirical conclusions of \cite{Oates6}.
We note however that validation of this inferred connectivity remains extremely challenging \citep[e.g.][]{Stein}.

\begin{figure}[t]
\centering
\includegraphics[width = \textwidth,clip,trim = 3.3cm 11.5cm 14.5cm 2.5cm]{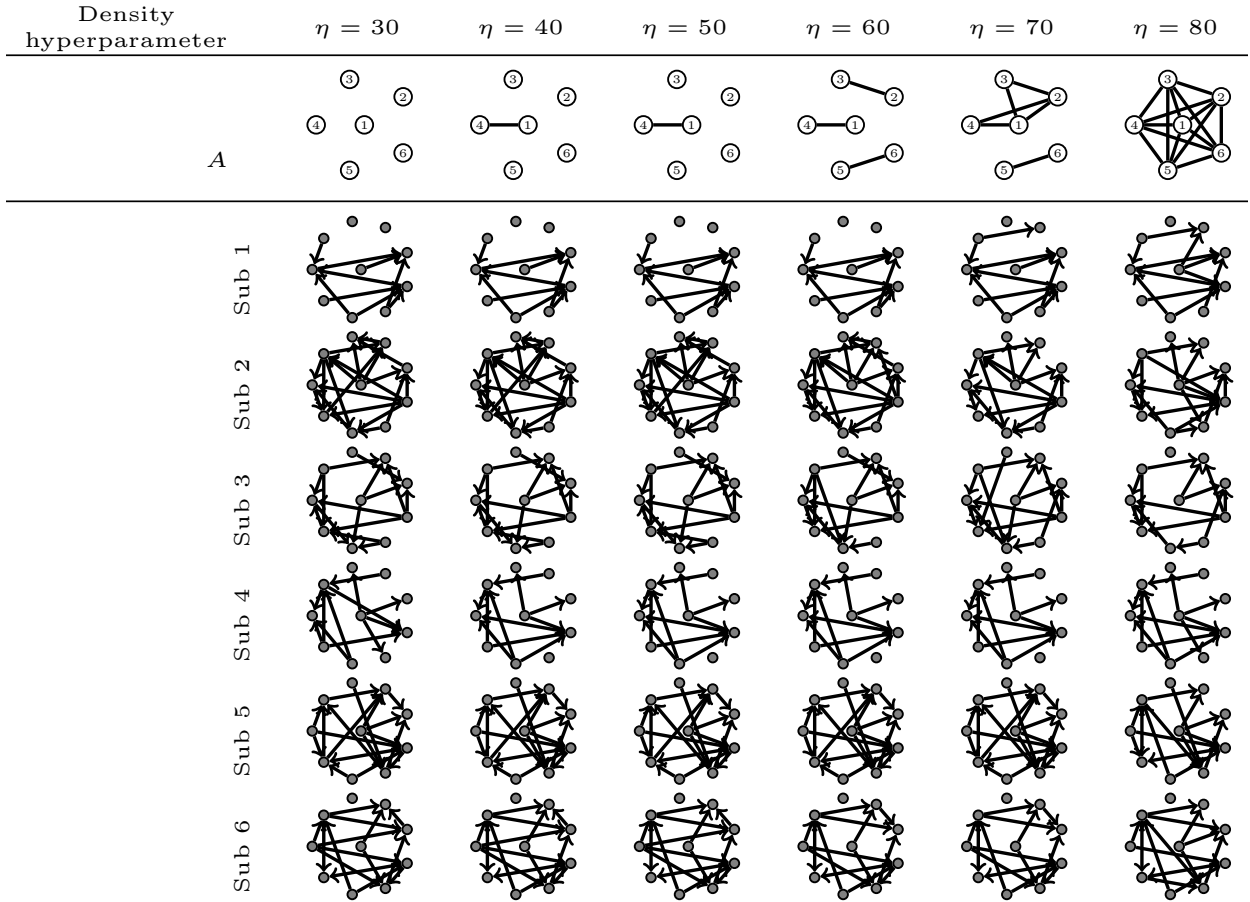}
\caption{Illustrative fMRI dataset; learning multiple DAGs {\it without} an exchangeability assumption. [Here we simultaneously estimate both subject-specific DAGs and the network $A$ that relates subjects. The regularity hyperparameter $\lambda = 4$ was fixed whilst the density hyperparameter $\eta$ was varied.]}
\label{N neuro}
\end{figure}

\subsection{Learning multiple DAGs without exchangeability}

\begin{figure}[t]
\centering
\begin{subfigure}[b]{0.49\textwidth}
\includegraphics[width = \textwidth,clip,trim = 0cm 6.5cm 0cm 0cm]{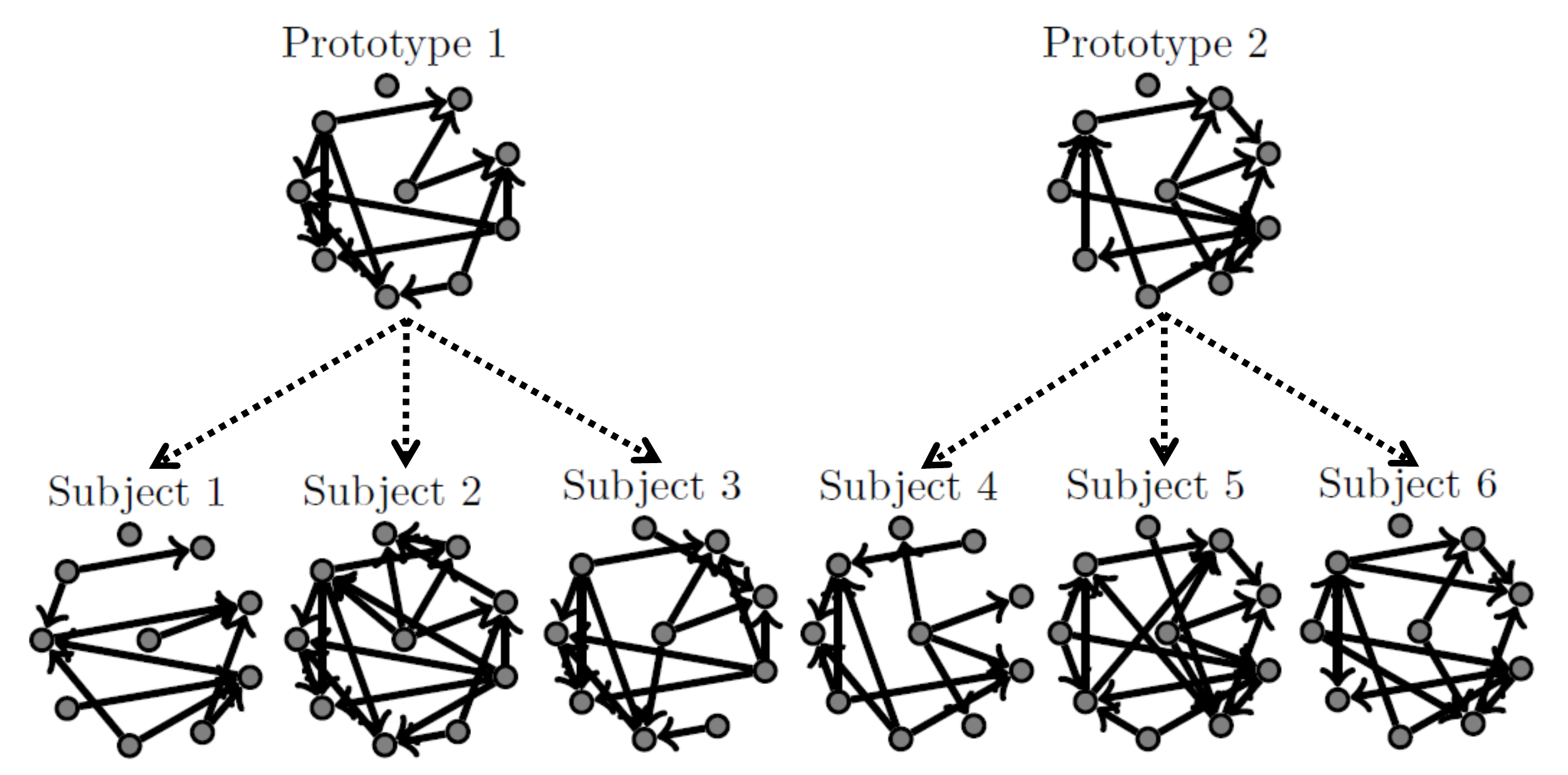} 
\caption{}
\label{km l2}
\end{subfigure}
\begin{subfigure}[b]{0.49\textwidth}
\includegraphics[width = \textwidth,clip,trim = 0cm 6.5cm 0cm 0cm]{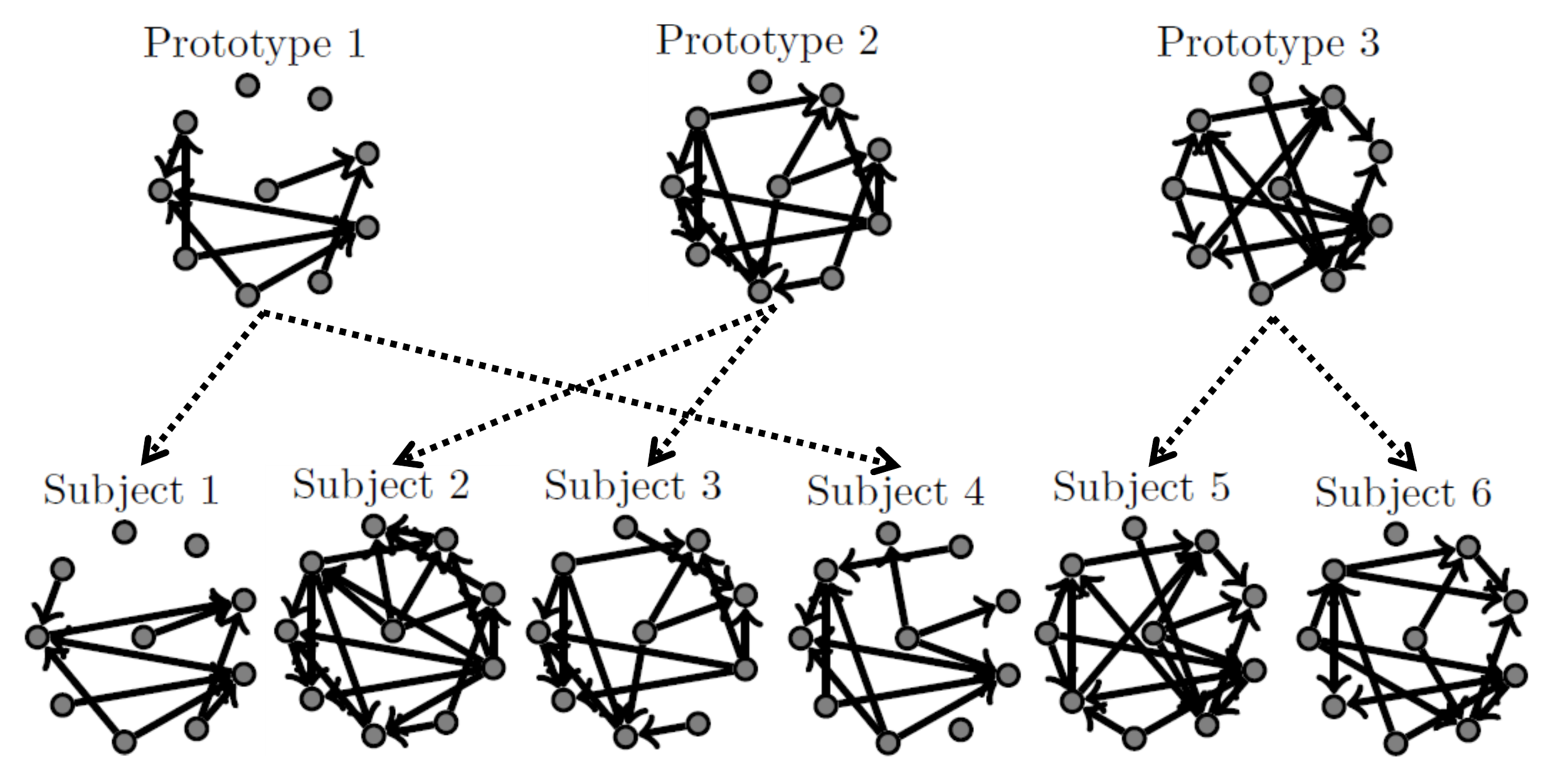} 
\caption{}
\label{km l3}
\end{subfigure}
\caption{Illustrative fMRI dataset; $k$-means clustering of DAGs with (a) $L=2$ clusters and (b) $L=3$ clusters. [We simultaneously estimate subject-specific DAGs, their cluster assignments (dashed edges) and the prototypes that summarise graphical structure within each cluster. The regularity hyperparameter was fixed at $\lambda = 4$.]}
\end{figure}

The scientific motivation for multi-subject analysis is typically to elucidate differential connectivity between subjects, either in a purely unsupervised context for exploratory investigation, or in a supervised context to determine whether certain features of connectivity are associated with auxiliary covariates of interest such as disease status.
In these cases a statistical model that assumes exchangeability between subjects may be inappropriate and ``regularise away'' the differential connectivity that is of interest.
We therefore proceed to jointly estimate both subject-specific DAGs $G^{(k)}$ and the network $A$ that describes relationships between the subjests themselves (Eqn. \ref{opt2}).

Elicitation of the hyperparameter $\eta$, that controls density of the network $A$, could again be performed by retrospective inspection of the posterior.
For our resting state fMRI dataset we would proceed by requiring (i) a moderate amount of similarity between subjects, motivated by expectation that connectivity should not differ substantially between subjects, and (ii) a moderate amount of heterogeneity between subjects, since we aim to highlight any potential differences between the neural connectivity of different subjects.
Results in Figure \ref{N neuro} demonstrate that for $\eta = 60$ the six subjects are regularised into three distinct components $\{1 ,4\}$, $\{2 , 3\}$, $\{5 , 6\}$, whilst for the higher value $\eta = 70$ the subjects are regularised into two distinct components $\{1,2,3,4\}$, $\{5,6\}$.
(When $\eta = 80$ the network $A$ is complete and subject-specific DAGs coincide with Fig. \ref{joint_neuro}.)
Examination of the Bayes factor as a function of $\eta$ demonstrates that the values $\eta = 60$, $70$ provide considerably better estimates compared to the DAGs obtained under an exchangeability assumption (log-Bayes factor $\approx 200$, $180$ respectively). 
This suggests that group and sub-group structure may be present amoung the subjects at the level of neural connectivity and provides evidence against exchangeability of the subjects.

Finally we illustrate an alternative and novel approach to learning similarities between subjects, called $k$-means clustering of DAGs, that does not assume exchangeability of the subjects.
In brief, additional latent DAGs $G^{(K+1)}, \dots, G^{(K+L)}$ are introduced that represent cluster centres or ``prototypes'', summarising the typical DAG structure within their cluster.
The (unknown) network $A$ on the extended vertex set $1,\dots,K+L$ is constrained to have edges that connect each of the vertices $1,\dots,K$ to precisely one of the vertices $K+1,\dots,K+L$, so that estimation of $A$ corresponds exactly to Bayesian model-based clustering.
Our methodology thereby facilitates joint estimation of both subject-specific DAGs and their optimal cluster assignment \citep{Oates6}.
(Note that, like in any mixture model, the optimal cluster assignment $A$ is defined only up to permutation of the cluster labels $K+1, \dots, K+L$.)
Here we applied $k$-means clustering of DAGs to the six subjects using $L=2$ clusters (Fig. \ref{km l2}) and $L=3$ clusters (Fig. \ref{km l3}).
The optimal cluster assignment with $L=3$ recovers the three distinct components $\{1 ,4\}$, $\{2 , 3\}$, $\{5 , 6\}$ that were obtained above via joint estimation of $A$, whilst the optimal cluster assignment with $L=2$ was $\{1,2,3\}$, $\{4,5,6\}$ which differs from the assignment obtained above in the position of the fourth subject only.
This analysis provides an alternative route to investigate similarity between the subjects and offers an alternative route to subjective elicitation of the $\eta$ hyperparameter.
We note that the prototypes that summarise cluster-specific graphical structure may be useful as summary statistics for the purposes of dimensionality reduction.

\FloatBarrier

\section{Discussion}

In neuroscience experiments it is increasingly common for data to be collected from multiple subjects whose neural connectivities are likely to be related but non-identical.
To uncover the causal mechanisms that underpin neural signalling it is necessary to work within a formal statistical theory for inferred causation, the most well-studied of which is rooted in DAGs \citep{Pearl}.
Yet until recently exact estimation for multiple related DAGs was computationally infeasible.
In this letter we have illustrated, using a small fMRI dataset, how recent algorithmic advances enable sophisticated causal inference using multi-subject experimental data.
In particular we have seen how novel statistical models, that do not assume exchangeability between subjects, achieve both a better description of the data (in terms of Bayes factors) and enable the more robust inference of subject-specific connectivity.

The model class that we discuss is large and allows for multiple opportunities to integrate prior knowlege, pertaining to both (i) the connectivity between specific neural regions, and (ii) additional covariates that might associate with subject-specific connectivity, such as age, gender or disease status.
The integration of auxiliary covariate data and the more general experimental validation of our techniques require an extensive and thorough investigation involving many more subjects than we analyse here; this is now the focus of our ongoing research. 

We focused on a particularly simple formulation with two tuning parameters and illustrated through application how both tuning parameters could be elicited retrospectively through examination of MAP estimates. 
This methodology extends naturally to highly structured datasets, for example where each subject is asked to provide multiple fMRI time courses. In these cases a combination of the techniques discussed above would permit all data on a particular subject to be aggregated into a single ``prototype'' and then estimation to proceed on the basis of these prototypes.

At present an analysis involving $K\leq 10$ subjects and DAGs of size $P \leq 10$ requires a few minutes' serial computation on a 2.10GHz Intel Core i7 CPU Windows host with 8GB memory. 
Our ongoing research focuses on reducing this computational burden so that exact estimation becomes feasible for much larger datasets that include hundreds of neural regions.
Recent advances in estimation of single DAGs involving thousands of nodes suggests that much progress can be made in this direction \citep{Barlett,Sheehan}.

Causal inference for neural connectivity is central to the study of brain functionality \citep{Smith2,Friston} and we envisage that the techniques presented here will play an important r$\hat{\text{o}}$le in the future analysis of multi-subject experimental data.

\section*{Appendix}

\begin{proof}[Proof of Property \ref{mix1}]
(a) This follows since when $\eta = 0$ the DAGs $G^{(1)}, \dots , G^{(k)}$ are {\it a priori} independent. Since the likelihood also factorises over $k$ it follows that the DAGs $G^{(1)}, \dots , G^{(k)}$ are independent in the posterior. 

(b) The objective that we wish to maximise can be written as 
\begin{eqnarray}
\bm{f}^T\bm{x} = \eta [(k,l) \in A] -\lambda \sum_{i=1}^P \sum_{j=1}^P [ (j \in G_i^{(k)}) \oplus (j \in G_i^{(l)}) ] \cap [(k,l) \in A] + C
\end{eqnarray} 
where $\eta > 0$ and $C$ does not depend on $[(k,l) \in A]$.
When $\hat{G}^{(k)} = \hat{G}^{(l)}$ it follows that the middle term is zero and hence to maximise this objective we must take $\hat{A}$ such that $[(k,l) \in \hat{A}] = 1$. 

(c,d) To prove both statements we can take 
\begin{eqnarray}
\lambda^* = \eta^* = \sum_{k=1}^K \sum_{i=1}^P \left[\max_{\pi \in \{1:P\}\setminus\{i\}} s^{(k)}(i,\pi)\right] - \left[\min_{\pi \in \{1:P\}\setminus\{i\}} s^{(k)}(i,\pi)\right].
\end{eqnarray}
For (c) note that if $\lambda > \lambda^*$ and $(k,l) \in A$, then the choice $G^{(k)} = G^{(l)}$ strictly maximises the objective function, since a selection $G^{(k)} \neq G^{(l)}$ incurs a penalty of at least $\lambda^*$ that cannot be compensated for by an increase in the likelihood term $\sum_{k=1}^K \sum_{i=1}^P \sum_{\pi \subseteq \{1:P\}\setminus\{ i \}} s^{(k)}(i,\pi) [G_i^{(k)} = \pi]$.
Similarly for (d), we have that $(k,l) \notin A$ incurs a penalty of at least $\eta^*$ that cannot be compensated for by an increase in the likelihood term. 
\end{proof} 

\begin{proof}[Proof of Property \ref{nonmono1}]
(a) Consider the following simple system with two variables and two individuals. 
Individual 1 has parent set scores $s^{(1)}(1,\{\}) = 0$, $s^{(1)}(1,\{2\})=-3$ for variable 1 and $s^{(1)}(2,\{\})=0$, $s^{(1)}(2,\{1\})=1$ for variable 2. 
Individual 2 has parent set scores $s^{(2)}(1,\{\})=0$, $s^{(2)}(1,\{2\})=4$ for variable 1 and $s^{(2)}(2,\{\})=0$, $s^{(2)}(2,\{1\})=1$ for variable 2. 
Then it is directly verified that for $0 \leq \lambda < 1$, $\hat{G}^{(1)} = 1 \rightarrow 2$ and $\hat{G}^{(2)} = 2 \rightarrow 1$, for $1 < \lambda < 2$, $\hat{G}^{(1)}$ has no edges and $\hat{G}^{(2)} = 2 \rightarrow 1$ and for $\lambda > 2$, $\hat{G}^{(1)} = 1 \rightarrow 2$ and $\hat{G}^{(2)} = 1 \rightarrow 2$. In particular, the edge $(1,2)$ is present in $\hat{G}^{(1)}$ for $\lambda \in [0,1) \cup (2,\infty)$ but absent for $\lambda \in (1,2)$. 

To embed the above example in a larger system with $P$ variables and $K$ individuals we proceed as follows:
Without loss of generality, assume $A(1,2) = 1$. For all variables in $G^{(1)}$ and $G^{(2)}$, assign scores $- \infty$ to any parent set $\pi$ that contains variables from both $\{1,2\}$ and $\{3,\dots,P\}$.
For variables $\{3, \dots, P\}$ in $G^{(1)}$ and $G^{(2)}$, and all variables in individuals $\{3,\dots,K\}$, take all scores to be zero (i.e. non-informative). 
Then the above proof demonstrates that the edge $(1,2)$ is present in $\hat{G}^{(1)}$ for $\lambda \in [0,1) \cup (2,\infty)$ but absent for $\lambda \in (1,2)$. 

(b)
Consider the following simple system with two variables and four individuals. 
Individual 1 has parent set scores $s^{(1)}(1,\{\})=0$, $s^{(1)}(1,\{2\})=0$ for variable 1 and $s^{(1)}(2,\{\})=0$, $s^{(1)}(2,\{1\})=1$ for variable 2.
Individual 2 has parent set scores $s^{(2)}(1,\{\})=0$, $s^{(2)}(1,\{2\})=0$ for variable 1 and $s^{(2)}(2,\{\})=0$, $s^{(2)}(2,\{1\})=2$ for variable 2. 
Individual 3 has parent set scores $s^{(3)}(1,\{\})=0$, $s^{(3)}(1,\{2\})=2$ for variable 1 and $s^{(3)}(2,\{\})=0$, $s^{(3)}(2,\{1\})=0$ for variable 2.  
Individual 4 has parent set scores $s^{(4)}(1,\{\})=0$, $s^{(4)}(1,\{2\})=3$ for variable 1 and $s^{(4)}(2,\{\})=0$, $s^{(4)}(2,\{1\})=0$ for variable 2.
Take $\lambda > \lambda^*$, as defined in Property \ref{mix1}, so that $\hat{G}^{(k)} = \hat{G}^{(l)}$ whenever $k \sim_{\hat{A}} l$.
Then it is directly verified that for $0 \leq \eta < 1$, $\hat{G}^{(1)} = 1 \rightarrow 2$, $\hat{G}^{(2)} = 1 \rightarrow 2$, $\hat{G}^{(3)} = 2 \rightarrow 1$, $\hat{G}^{(4)} = 2 \rightarrow 1$, and $\hat{A} = \{(1,2),(3,4)\}$; for $1 < \eta < 3/2$, $\hat{G}^{(1)} = 2 \rightarrow 1$, $\hat{G}^{(2)} = 1 \rightarrow 2$, $\hat{G}^{(3)} = 2 \rightarrow 1$, $\hat{G}^{(4)} = 2 \rightarrow 1$, and $\hat{A} = \{(1,3),(3,4),(4,1)\}$; for $\eta > 3/2$,  $\hat{G}^{(1)} = 2 \rightarrow 1$, $\hat{G}^{(2)} = 2 \rightarrow 1$, $\hat{G}^{(3)} = 2 \rightarrow 1$, $\hat{G}^{(4)} = 2 \rightarrow 1$, and $\hat{A}$ is the complete network.
In particular, the edge $(1,2)$ is present in $\hat{A}$ for $\lambda \in [0,1) \cup (3/2,\infty)$ but absent for $\lambda \in (1,3/2)$. 
\end{proof}

\section*{Acknowledgements}
CJO is supported by the Centre for Research in Statistical Methodology (CRiSM) UK EPSRC EP/D002060/1.
LC is supported by Coordena\c{c}\~ao de Aperfei\c{c}oamento de Pessoal de N\'ivel Superior (CAPES), Brazil.
TEN is supported by the Wellcome Trust, 100309/Z/12/Z and 098369/Z/12/Z, and by NIH grants U54 MH091657-03, R01 NS075066-01A1 and R01 EB015611-01.
The authors are grateful to James Cussens, Jim Smith and Sach Mukherjee for many helpful discussions on the methodology that is presented here, and to Stephen Smith for the preprocessing and preparation of fMRI data.

\end{document}